\documentclass{article}
\pdfoutput=1
\usepackage[margin = 1in]{geometry}
 \usepackage{amsmath,amsfonts,amsthm}
\usepackage{microtype}
\usepackage{hyperref}

\usepackage[sort]{cite}

\usepackage{hyperref,xcolor}
\hypersetup{
     colorlinks   = true,
     citecolor    = blue
}

\theoremstyle{plain}
\newtheorem{theorem}{Theorem}
\newtheorem{definition}[theorem]{Definition}

\newtheorem{corollary}[theorem]{Corollary}
\newtheorem{lemma}[theorem]{Lemma}

\usepackage{url}
\usepackage{color,microtype}
\usepackage{graphicx}
\usepackage{amssymb,amsfonts}

\usepackage{textcomp}
\usepackage{multirow}
\usepackage{gensymb}
\usepackage{array}

\usepackage{xspace}

\newcommand{\expec}[1]{ \textup{E}\left [ #1 \right ]}

\newcommand{\roundup}[1]{\left \lceil  #1 \right \rceil}

\newcommand{\expo}[1]{ \exp \left ( #1 \right )}
\newcommand{\size}[1]{\left | #1 \right |}
\newcommand{\paren}[1]{\left ( #1 \right )}

\DeclareMathOperator{\polylog}{polylog}
\DeclareMathOperator{\poly}{poly}

\DeclareMathOperator{\id}{{ID}}

\DeclareMathOperator{\cover}{{\mathcal{C}}}

\newcommand{\tO}{\tilde{O}}

\newcommand{\coverage}{\cover}

\newcommand{\prob}[1]{ \Pr \left [ #1 \right ]}

\renewcommand{\O}{{\cal O}}
\newcommand{\etal}{{et al.~}}

\newcommand{\eat}[1]{}

\newcommand{\W}{{\cal W}}

\renewcommand{\S}{{\cal S}}

\newcommand{\MkVC}{$\mathsf{MaxVertexCoverage}$\xspace}
\newcommand{\MkSC}{$\mathsf{MaxSetCoverage}$\xspace}

\newcommand{\Disj}{\mathsf{DISJ}\xspace}

\usepackage{algorithmic}  
\usepackage[boxed]{algorithm2e}
\SetArgSty{textrm}

\DeclareMathOperator{\opt}{OPT}

\usepackage{libertine,newtxmath}

\title{Better Streaming Algorithms for the Maximum Coverage Problem}

\author{Andrew McGregor\thanks{This work was supported by NSF  Awards CCF-0953754, IIS-1251110,  CCF-1320719, and a Google Research Award.} \\ University of Massachusetts \\ mcgregor@cs.umass.edu  \and Hoa T. Vu\footnotemark[1] \\ University of Massachusetts \\ hvu@cs.umass.edu}

\date{}

\begin{document}

\begin{titlepage}
\maketitle


\begin{abstract} 
 We study the classic NP-Hard problem of finding the maximum $k$-set coverage in the data stream model: given a set system of $m$ sets that are subsets of a universe $\{1,\ldots,n \}$, find the $k$ sets that cover the most number of distinct elements. The problem can be approximated up to a factor $1-1/e$ in polynomial time. In the streaming-set model, the sets and their elements are revealed online. The main goal of our work is to design algorithms, with approximation guarantees as close as possible to $1-1/e$, that use sublinear space $o(mn)$. Our main results are:
\begin{itemize} 
\item Two $(1-1/e-\epsilon)$ approximation algorithms: One uses $O(\epsilon^{-1})$ passes and $\tO(\epsilon^{-2} k)$ space whereas the other uses only a single pass but $\tO(\epsilon^{-2} m)$ space. $\tO(\cdot)$ suppresses $\polylog$ factors.
\item We show that any approximation factor better than $(1-(1-1/k)^k)\approx 1-1/e$  in constant passes requires $\Omega(m)$ space for constant $k$ even if the algorithm is allowed unbounded processing time. We also demonstrate a  \emph{single-pass}, $(1-\epsilon)$ approximation algorithm using $\tO(\epsilon^{-2} m \cdot \min(k,\epsilon^{-1}))$ space.
\end{itemize} 
We also study the maximum $k$-vertex coverage problem in the dynamic graph stream model. In this model, the stream consists of edge insertions and deletions of a graph on $N$ vertices. The goal is to find $k$ vertices that cover the most number of distinct edges. 
\begin{itemize} 
\item We show that any constant approximation in constant passes requires $\Omega(N)$ space for constant $k$ whereas $\tO(\epsilon^{-2}N)$ space is sufficient for a $(1-\epsilon)$ approximation and arbitrary $k$ in a single pass.
\item For regular graphs, we show that $\tO(\epsilon^{-3}k)$ space is sufficient for a $(1-\epsilon)$ approximation in a single pass. We generalize this to a $(\kappa-\epsilon)$ approximation when the ratio between the minimum  and maximum degree is bounded below by $\kappa$. 
\end{itemize} 
\end{abstract} 
\end{titlepage}
\section{Introduction}

The \emph{maximum set coverage problem} is a classic NP-Hard problem that has a wide range of applications including facility and sensor allocation \cite{KrauseG07}, information retrieval \cite{Anagnostopoulos15}, influence maximization in marketing strategy design \cite{KempeKT15}, and the blog monitoring problem where we want to choose a small number of blogs that cover a wide range of topics \cite{SahaG09}. In this problem, we are given a set system of $m$ sets that are subsets of a universe $[n]:=\{1,\ldots,n \}$. The goal is to find the $k$ sets whose union covers the largest number of distinct elements. For example, in the application considered by Saha and Getoor \cite{SahaG09}, the universe corresponds to $n$ topics of interest to a reader, each subset corresponds to a blog that covers some of these topics, and the goal is to maximize the number of topics that the reader learns about if she can only choose $k$ blogs.

It is well-known that the greedy algorithm, which greedily picks the set that covers the most number of uncovered elements, is a $1-1/e$ approximation. Furthermore, unless $P=NP$, this approximation factor is the best possible \cite{Feige98}. 

The \emph{maximum vertex coverage problem} is a special case of this problem in which the universe corresponds to the edges of a given graph and there is a set corresponding to each  node of the graph that contains  the edges that are incident to that node. For this problem, algorithms based on linear programming achieve a $3/4$ approximation for general graphs  \cite{AgeevS99} and a $8/9$ approximation for bipartite graphs \cite{CaskurluMPS14}.  Assuming $P\neq NP$, there does not exist a polynomial-time approximation scheme. Recent work has focused on finding purely combinatorial algorithms for this problem \cite{BonnetEPS16}. 
\paragraph{Streaming Algorithms.}
Unfortunately, for both problems, the aforementioned greedy and  linear programming algorithms scale poorly to massive data sets. This  has motivated a significant research effort in designing algorithms that could handle large data in modern computation models such as the data stream model and the MapReduce model \cite{KumarMVV15,BadanidiyuruMKK14}. In the data stream model, the $k$-set coverage problem and the related set cover problem have received a lot of attention in recent research \cite{Har-PeledIMV16,ChakrabartiW16,AusielloBGLP12, YuY13,EmekR14,AssadiKL16}.

Two variants of the data stream model are relevant to our work. In the \emph{streaming-set model} \cite{SahaG09,GuhaMT15,EmekR14,RadhakrishnanS11,Sun13,KoganK14}, the stream consists of $m$ sets $S_1,\ldots,S_m$ and each $S_i$ is encoded as the list of elements in that set along with a unique ID for the set. For simplicity, we assume that $\id(S_i)=i$. In the dynamic graph stream model \cite{AhnCGMW15,AhnGM12a,AhnGM12b,AhnGM13,KapralovLMMS14,KapralovW14,GuhaMT15,BhattacharyaHNT15,ChitnisCEHMMV16,AssadiKLY16,Konrad15,McGregorVV16,McGregorTVV15}, relevant to the maximum vertex coverage problem, the stream consists of insertions and deletions of edges of the underlying graph. For a recent survey of research in graph streaming, see \cite{McGregor14}.  Note that any algorithm for the dynamic graph stream model can also be used in the streaming-set model; the streaming-set model is simply a special case in which there is no deletion and edges are grouped by endpoint.

\subsection{Related Work}
\paragraph{Maximum Set Coverage.} 
Saha and Getoor \cite{SahaG09} gave a swap based $1/4$ approximation  algorithm that uses a single pass and $\tO(kn)$ space. At any point, their algorithm stores $k$ sets explicitly in the memory as the current solution. When a new set arrives, based on a specific rule, their algorithm either swaps it with the set with the least contribution in the current solution or does nothing and moves on to the next set in the stream. Subsequently, Ausiello \etal \cite{AusielloBGLP12} gave a slightly different swap based algorithm that also finds a $1/4$ approximation using one pass and the same space. Yu and Yuan \cite{YuY13} claimed an $\tO(n)$ space, single-pass algorithm with an approximation factor around 0.3 based on the aid of computer simulation.

Recently, Badanidiyuru \etal \cite{BadanidiyuruMKK14} gave a generic single-pass algorithm for maximizing a monotone submodular function on the stream's items subject to the cardinality constraint that at most $k$ objects are selected. Their algorithm guarantees a $1/2-\epsilon$ approximation. At a high level, based on a rule that is different from \cite{SahaG09,AusielloBGLP12} and a guess of the optimal value, their algorithm decides if the next item (which is a set in our case) is added to the current solution. The algorithm stops when it reaches the end of the stream or when $k$ items have been added to the solution. In the maximum  set coverage problem, the rule requires knowing the coverage of the current solution.  As a result, a careful adaptation to the maximum set coverage problem uses $\tO(\epsilon^{-1} n)$ space. For constant $\epsilon$, this result directly improves upon \cite{SahaG09,AusielloBGLP12}. Subsequently, Chekuri \etal \cite{ChekuriGQ15} extended this work to non-monotone submodular function maximization under
constraints beyond cardinality.

The set cover problem, which is closely related to the maximum set coverage problem, has been studied in \cite{SahaG09,Har-PeledIMV16,ChakrabartiW16,EmekR14,AssadiKL16}. See \cite{AssadiKL16} for a comprehensive summary of results and discussion.

\paragraph{Maximum Vertex Coverage.} 
The streaming maximum vertex coverage problem  was studied by Ausiello et al.~\cite{AusielloBGLP12}. They first observed that simply outputting the $k$ vertices with highest degrees is a 1/2 approximation; this can easily be done in the streaming-set model. The main results of their work were  $\tO(kN)$-space  algorithms that have better approximation for special types of graph. Their results include a 0.55 approximation for regular graphs and a 0.6075 approximation for regular bipartite graphs. Note that their paper only considered the streaming-set model whereas our results for maximum vertex coverage will consider the more challenging dynamic graph stream model.

\subsection{Our Contributions}
\paragraph{Maximum $k$-set coverage.} Our main goal is to achieve the $1-1/e$ approximation that is possible in the non-streaming or offline setting. 
\begin{itemize}
\item We present polynomial time data stream algorithms that achieve a $1-1/e-\epsilon$ approximation for arbitrarily small $\epsilon$. The first algorithm uses  one pass and $\tO(\epsilon^{-2} m)$ space whereas the second algorithm  uses $O(\epsilon^{-1})$ passes and $\tO(\epsilon^{-2} k)$ space. We consider both algorithms to be pass efficient but the second algorithm uses much less space at the cost of using more than one pass. We note  that storing the solution itself requires $\Omega(k)$ space. Thus, we consider $\tO(\epsilon^{-2} k)$ space to be surprisingly space efficient.
\item For constant $k$, we show that $\Omega(m)$ space is required by any constant pass (randomized) algorithm to achieve an approximation factor better than $1-(1-1/k)^k$ with probability at least 0.99; this holds even if the algorithm is permitted exponential time. To the best of our knowledge, this is the first non-trivial space lower bound for this problem. However, with exponential time and $\tO(\epsilon^{-2} m \cdot \min(k,\epsilon^{-1}))$ space we observe that a $1-\epsilon$ approximation is possible in a single pass.
\end{itemize}

For a slightly worse approximation, a $1/2-\epsilon$ approximation in  one pass can be achieved using $\tO(\epsilon^{-3} k)$ space. This follows by building on the result of Badanidiyuru \etal \cite{BadanidiyuruMKK14}. However, we provide a simpler algorithm and analysis. 

Our approach generalizes to the group cardinality constraint in which there are $\ell$ groups and only $k_i$ sets from group $i$ can be selected. This is also known as the partition matroid constraint. We give a $1/(\ell+1)-\epsilon$ approximation which improves upon \cite{ChakrabartiK14,ChekuriGQ15} for the case $\ell = 2$. Let $k = k_1 + \ldots + k_\ell$. If $O(\epsilon^{-1} \log (k/\epsilon))$ passes are permitted, then we could achieve a $1/2-\epsilon$ approximation by adapting the greedy analysis in  \cite{ChekuriK04} to our framework.

Finally, we design a $1/3-\epsilon$ approximation algorithm for the budgeted maximum set coverage problem using one pass and $\tO(\epsilon^{-1}(n+m))$ space. In this version, each set $S$ has a cost $w_S$ in the range $[0,L]$. The goal is to find a collection of sets whose total cost does not exceed $L$ that cover the most number of distinct elements. Khuller \etal \cite{KhullerMN99} presented a polynomial time and $1-1/e$ approximation algorithm based on the greedy algorithm and an enumeration technique. Our results are summarized in Figure \ref{fig:results1}.

Shortly after our original submission, in an independent work, Bateni \etal \cite{BateniEM16} also presented a polynomial-time, single-pass, $\tO(\epsilon^{-3} m)$ space algorithm that finds a $1-1/e-\epsilon$ approximation for the maximum $k$-set coverage problem. Furthermore, given unlimited post-processing time, their results also imply a $1-\epsilon$ approximation using  a single-pass and $\tO(\epsilon^{-3} m)$ space. This extension to $1-\epsilon$ approximation is also possible with our approach; see the end of \ref{sec:mspace} for details. We also note that our approach also works in their \emph{edge arrival} model in which the stream reveals the set-element relationships one at a time. 

Recently, Assadi proved a space lower bound $\Omega(\epsilon^{-2} m)$ for any $1-\epsilon$ approximation for constant $k$ \cite{Assadi17}.

\paragraph{Maximum $k$-vertex coverage.} Compared to the most relevant previous work \cite{AusielloBGLP12}, we study this problem in a more general model, i.e., the dynamic graph stream model. We manage to achieve a better approximation and space complexity for general graphs even when comparing to their results for special types of graph. Our results are summarized in Figure \ref{fig:results2}.
 In particular, we show that
\begin{itemize}
\item $\tO(\epsilon^{-2}N)$ space is sufficient for a $1-\epsilon$ approximation (or a $3/4-\epsilon$ approximation if  restricted to polynomial time) and arbitrary $k$ in a single pass. The algorithms in  \cite{AusielloBGLP12} use $\tO(kN)$ space and achieve an approximation worse than 0.61 even for special graphs.
\item Any constant approximation in constant passes requires $\Omega(N)$ space for constant $k$.
\item For regular graphs, we show that $\tO(\epsilon^{-3}k)$ space is sufficient for $1-\epsilon$ approximation in a single pass. We generalize this to an $\kappa -\epsilon$ approximation when the ratio between the minimum and maximum degree is bounded below by $\kappa$. We also extend this result to hypergraphs.
\end{itemize}

\begin{figure}[t]
\centering
\begin{tabular}{ c c c c c  c } 
\hline\noalign{\smallskip}
Theorem & Bound & No. of passes & Space & Approximation & Constraint  \\
\noalign{\smallskip}\hline\noalign{\smallskip}
 \ref{thm:key-result-1} & U & $O\left( \epsilon^{-1} \right)$ & $\tO(\epsilon^{-2}{k}{})$ & $1-1/e -\epsilon$ & C  \\
 \ref{thm:key-result-3} & U & 1 & $\tO(\epsilon^{-3}{k}{})$ & $1/2 - \epsilon$ & C \\
 \ref{thm:key-result-2} & U & 1 & $\tO(\epsilon^{-2}{m}{})$ & $1-1/e-\epsilon$ & C \\
 \ref{thm:key-result-4}, \ref{thm:key-result-5} & U &1 & $\tO(\epsilon^{-2}{m \cdot \min(k, \epsilon^{-1}}))$ & $1 - \epsilon$ & C \\
 \ref{thm:lower-bound-1} & L & Constant & $\Omega({m}{k^{-2}})$ & $1-(1-1/k)^k+\epsilon$ & C \\
 \ref{thm:sub-result-1} & U & 1 & $\tO(\epsilon^{-3}k)$ & $1/(\ell+1)-\epsilon$ & G \\
 \ref{thm:sub-result-2} & U & $O(\epsilon^{-1} \log (k/\epsilon))$ & $\tO(\epsilon^{-2}k)$ & $1/2-\epsilon$ & G \\
  \ref{theorem:budget} & U& 1 & $\tO(\epsilon^{-1}{(n+m)}{})$ & $1/3-\epsilon$ & B \\
\noalign{\smallskip}\hline
\end{tabular}
\centering
\caption{Summary of results for \MkSC,    
U: upper bound, L: lower bound, C: cardinality, B: budget, G: group cardinality (partition matroid)  }\label{fig:results1}
\end{figure}

\begin{figure}[t]
\centering
\begin{tabular}{c c c c c} 
\hline\noalign{\smallskip}
Theorem & Bound  & No. of passes & Space & Approximation   \\
\noalign{\smallskip}\hline\noalign{\smallskip}
 \ref{thm:vertex-cover-general-2} & U & 1 & $\tO(\epsilon^{-2}{N}{})$ & $1-\epsilon$   \\
 \ref{thm:vertex-cover-bounded-2} & U & 1 & $\tO(\epsilon^{-3}{k}{})$ & $\kappa -\epsilon$    \\
 \ref{thm:lower-bound-2}& L & 1 & $\Omega(N \kappa^3/k)$ & $\kappa + \epsilon$    \\
\noalign{\smallskip}\hline
\end{tabular}
\caption{Summary of results for \MkVC, U: upper bound, L: lower bound, $\kappa$ is ratio of lowest degree to highest degree.}
 \label{fig:results2}
\end{figure}

\paragraph{Our techniques.} 
On the algorithmic side, our basic approach is a  ``guess, subsample, and verify'' framework. At a high level, suppose we design a streaming algorithm for approximating the maximum  $k$-set coverage that assumes a priori knowledge of a good guess of the optimal coverage. We show that it is a) possible to run same algorithm on a subsampled universe  defined by a carefully chosen hash function and b) remove the assumption that a good guess was already known.

If the guess is at least nearly correct, running the algorithm on the subsampled universe results in a small space complexity. However, there are two main challenges. First, an algorithm instance with a wrong guess could use too much space. We simply terminate those instances. The second issue is more subtle. Because the hash function is not fully independent, we appeal to a special version of the Chernoff bound. The bound needs not guarantee a good approximation unless the guess is near-correct. To this end, we use the $F_0$ estimation algorithm to verify the coverage of the solutions. Finally, we return the solution with maximum estimate coverage. This framework allows us to restrict the analysis solely to the near-correct guess. The analysis is, therefore, significantly simpler.

Some of our other algorithmic ideas are inspired by previous works. The ``thresholding greedy'' technique was inspired by \cite{ChakrabartiW16,CormodeKW10,BadanidiyuruV14}. However, the analysis is different for our problem. Furthermore, to optimize the number of passes, we rely on new observations.

Another algorithmic idea in designing one-pass space-efficient algorithm is to treat the sets differently based on their contributions. During the stream, we immediately add the sets with large contributions to the solution. We store the contribution of each remaining sets explicitly and solve the remaining problem offline. Har-Peled \etal \cite{Har-PeledIMV16} devised a somewhat similar strategy but the details are  different.  

For the maximum $k$-vertex coverage problem, we show that simply running the streaming cut-sparsifier algorithm is sufficient and optimal up to a polylog factor. The novelty is to treat it as an interesting corner case of a more space-efficient algorithm for near regular graphs, i.e., $\kappa$ is bounded below.

One of the novelties is  proving the lower bound via a randomized reduction from the $k$-party set disjointness problem. 

\paragraph{Comparison to the conference publication.} This is an extended and revised version of a preliminary version in ICDT 2017 \cite{V17}. In this version, we present the single-pass, $\tO(\epsilon^{-3}m)$ space and $1-\epsilon$ approximation algorithm that was briefly mentioned in \cite{V17}. Furthermore, we provide two new algorithms for the maximum set coverage problem under the group cardinality constraint in Section \ref{sec:group-cardinality}.


\section{Algorithms for maximum $k$-set coverage}

In this section, we design various algorithms for approximating \MkSC in the data stream model. Our main algorithmic results in this section are two $1-1/e-\epsilon$ approximation algorithms. The first algorithm uses  one pass and $\tO(\epsilon^{-2} m)$ space whereas the second algorithm  uses $O(\epsilon^{-1})$ passes and $\tO(\epsilon^{-2} k)$ space.  We also briefly explore some other trade-offs in a subsequent subsection.
 
\paragraph{Notation.}
If $\mathcal{A}$ is a collection of sets, then $\cover(\mathcal{A})$ denotes the union of these sets.

\subsection{$(1-1/e-\epsilon)$ approximation in one pass and $\tO(\epsilon^{-2}m)$ space} 
\label{sec:mspace}

\paragraph{Approach.} 
The algorithm  adds sets to the current solution if the number of new elements they cover exceeds some threshold. The basic algorithm relies on an estimate $z$ of the optimal coverage $\opt$. The threshold for including a new set in the solution is that it covers at least $z/k$ new elements. Unfortunately, this threshold is too high to ensure that we selected sets that achieve the required $1-1/e-\epsilon$ approximation  and we may want to revisit adding a set, say $S$, that was not added when it first arrived. To facilitate this, we will explicitly store the subset of $S$ that were uncovered when $S$ arrived in a collection of sets $\mathcal{W}$. Because $S$ was not added immediately, we know that this subset is not too large. At the end of the pass, we continue augmenting out current solutions using the collection $\mathcal{W}$.

 \paragraph{Technical Details.} 
For the time being, we suppose that the algorithm is provided with an estimate $z$ such that $\opt \leq z \leq 4 \opt$. We will later remove this assumption. The algorithm uses $C$ to keep track of the elements that have been covered so far. Upon seeing a new set $S$, the algorithm stores $S \setminus C$ explicitly in $\mathcal{W}$ if $S$ covers few new elements. Otherwise, the algorithm adds $S$ to the solution and updates $C$ immediately. At the end of the stream, if there are fewer than $k$ sets in the solution, we use the greedy approach to find the remaining sets from  $\mathcal{W}$ .

The basic algorithm maintains $I\subseteq [m]$, $C\subseteq [n]$ where $I$ corresponds to the ID's of the (at most $k$) sets in the current solution and $C$ is the the union of the corresponding sets. We also maintain a collection of sets $\W$ described above. 
  The algorithm proceeds as follows:

\begin{enumerate}
\item Initialize $C=\emptyset$, $I=\emptyset$, $\W=\emptyset$.
\item For each set $S$ in the stream: 
\begin{enumerate}
\item If $|S\setminus C| <  {z/k}$ then $\W \leftarrow \W \cup \{S\setminus C\}$.
\item If $|S\setminus C| \geq  {z/k}$ and $\size{I} < k$, then $I\leftarrow I\cup \{ID(S)\}$ and $C\leftarrow C\cup S$.
\end{enumerate}
\item Post-processing: Greedily add $k-|I|$ sets from $\W$ and update $I$ and $C$ appropriately.
\end{enumerate}

\begin{lemma}\label{lem:core-2}
There exists a single-pass, $O\left(k\log m +{mz}/{k} \cdot\log n \right)$-space algorithm that finds a $1-1/e$ approximation of \MkSC. 
\end{lemma}

\begin{proof}
We observe that storing the set of covered elements $C$  requires at most $\opt \log n= O( z \log n)$ bits of space. For each set $S$ such that ${S}\setminus C$ is stored explicitly in $\mathcal{W}$, we need $O\left( z/k \cdot  \log n \right)$ bits of space. Storing $I$ requires $O(k\log m)$ space.
Thus, the  algorithm uses the space as claimed.

After the algorithm added the $i$th set $S$ to the solution, let $a_i$ be the number of new elements that $S$ covers and $b_i$ be the total number of covered elements so far. Furthermore, for $i>0$, let $c_i = \opt - b_i$. Define $a_0 := b_0 := 0$ and $c_0 := \opt$. 

At the end of the stream, suppose $\size{I} = j$. Then, 
\[
c_j \leq  \opt-\frac{z \cdot j}{k}\leq \opt \paren{1-\frac{j}{k}}\leq \opt \paren{1-\frac{1}{k}}^j~.\] 
The last inequality holds when $k \geq 2$ and $j$ is a non-negative integer. The case $k=1$ is trivial since we can simply find the largest set in $\tO(1)$ space. 

Now, we consider the sets that were added  in post-processing. We then proceed with the usual inductive argument to show that $c_i \leq (1-1/k)^i \opt$ for $i>j$. Before the algorithm added the $(i+1)$th set for $j \leq i \leq k-1$, there must be a set that covers at least ${c_i}/{k} $ new elements. Therefore, 
\[c_{i+1}=c_i - a_{i+1} \leq c_i \paren{ 1-\frac{1}{k} } \leq \opt \paren{1-\frac{1}{k}}^{i+1} \ .\] 
The approximation follows since $c_k \leq \opt(1-1/k)^{k} \leq 1/e \cdot \opt$.
\end{proof}

Following the approach outlined in Section \ref{sec:magic} we may assume $z=O(\epsilon^{-2} k \log m)$ and that $\opt \leq z\leq 4\opt.$

\begin{theorem}\label{thm:key-result-2}
There exists a single-pass, $\tO(\epsilon^{-2} m)$ space algorithm that finds a $1-1/e-\epsilon$ approximation of  \MkSC with high probability.
\end{theorem}

\paragraph{Better approximation using more space and unlimited post-processing time.} We observe that a slight modification of the above algorithm can be used to attain a $1-1/(4b)$ approximation for any $b>1$ if we are permitted unlimited post-processing time and an  extra factor of $b$ in the space use. Specifically, we increase the threshold for when to add a set immediately to the solution from $z/k$ to $bz/k$ and then find the optimal collection of $k-|I|$ sets from $\W$ to add in post-processing.  It is immediate that this algorithm uses $O(k\log m + mbz/k \cdot \log n)$ space. 

Suppose a collection of $y$ sets $\S_1$ were added during the stream. These $y$ sets  cover 
\[|\cover(\S_1)|  \geq y \cdot \frac{bz}{k} \geq \opt \cdot \frac{yb}{k}  \] elements. On the other hand, the collection of sets $\S_2$ selected in post-processing covers at least $\frac{k-y}{k} \cdot  \left( \opt - |\cover(\S_1)|  \right)$ new elements. Then,
\begin{align*}
\size{\cover(\S_1\cup \S_2) }
& \geq |\cover(\S_1)| + \frac{k-y}{k} \cdot  \left( \opt -| \cover(\S_1)|  \right) \\
& = \paren{1-\frac{y}{k}} \opt + \frac{y}{k} \cdot |\cover(\S_1) | \\
& \geq  \paren{1-\frac{y}{k}} \opt + \frac{y}{k} \cdot \opt \cdot \frac{yb}{k}   \\
& = \opt \left( 1 - \frac{y}{k} + \left( \frac{y}{k}\right)^2 b\right) \\
& \geq \opt \left( 1-\frac{1}{4b} \right) 
\end{align*}  
where the last inequality follows by minimizing over $y$.
Hence, we obtain a $1-\epsilon$ approximation by setting $b = 4/\epsilon$.

\begin{theorem}\label{thm:key-result-4}
There exists a single-pass, $\tO(\epsilon^{-3} m)$ space algorithm that finds a $1-\epsilon$ approximation of  \MkSC with high probability.
\end{theorem}

\subsection{$(1-1/e-\epsilon)$ approximation in $O(\epsilon^{-1})$ \label{sec:thesholding-greedy} passes and $\tO(\epsilon^{-2} k)$ space} 

\label{sec:multiple-pass}

\paragraph{Approach.}
Our second algorithm is based on the standard greedy approach but instead of adding the set that increases the coverage of the current solution  the most at each set, we add a set if the number of new elements covered by this set exceeds a certain threshold. This threshold decreases with each pass in such a way that after only $O(\epsilon^{-1})$ passes, we have a good approximate solution but the resulting algorithm may use too much space.  
We will fix this by first randomly subsampling each set at different rates and running multiple instantiations of the basic algorithm corresponding to different rates of subsampling.

The basic ``decreasing threshold'' approach has been used before in different contexts  \cite{BadanidiyuruV14,ChakrabartiW16,CormodeKW10}. The novelty of our approach is in implementing this approach such that the resulting algorithm uses small space and a small number of passes. For example, a direct implementation of the approach by Badanidiyuru and  Vondr\'ak \cite{BadanidiyuruV14} in the streaming model may require $O(\epsilon^{-1} \log (m/\epsilon))$ passes and $O(n)$ space\footnote{Note that their work addressed the more general problem of maximizing sub-modular functions.}.

\paragraph{Technical Details.}
We will assume that we are given an estimate $z$ of $\opt$ such that $\opt\leq z\leq 4\opt$. We start by designing a  $(1-1/e-\epsilon$) approximation algorithm that uses $\tO(k+z)$ space and $O({\epsilon^{-1}})$ passes. 
We will subsequently use a sampling approach to reduce the space to $\tO(\epsilon^{-2}k)$.

As with the previous algorithm, the basic algorithm in this section also  maintains $I\subseteq [m]$, $C\subseteq [n]$ where $I$ corresponds to the ID's of the (at most $k$) sets in the current solution and $C$ is the the union of the corresponding sets. The algorithm proceeds as follows:

\begin{enumerate}
\item Initialize $C=\emptyset$ and $I=\emptyset$
\item For $j=1$ to $1+\roundup{\log_{\alpha} (4e)}$ where $\alpha=1+\epsilon$: 
\begin{enumerate}
\item Make a pass over the stream. For each set $S$ in the stream: If $|I|<k$ and 
\[
|S\setminus C| \geq \frac{z}{k(1+\epsilon)^{j-1}}~,
\] 
then $I\leftarrow I\cup \{ID(S)\} ~~\mbox{ and }~~ C\leftarrow C\cup S \ . $ 
\end{enumerate}
\end{enumerate}

\begin{lemma}\label{lem:core-1}
There exists an  $O(\epsilon^{-1})$-pass, $O(k \log m+z \log n)$-space algorithm  that finds a $1- 1/e-\epsilon$ approximation of \MkSC. 
\end{lemma}

To analyze the algorithm, we introduce some notation. After the $i$th set was picked, let $a_i$ be the number of new elements  covered by this set and let $b_i$ be the total number of covered elements so far. Furthermore, let $c_i = \opt -b_i$. We define $a_0:=0$ and $b_0:=0$. 
\begin{lemma}\label{lem:3}
Suppose the algorithm picks $k'$ sets. For $0\leq i\leq k'-1$,  we have $a_{i+1} \geq {c_i}/({\alpha k})$.
\end{lemma}

\begin{proof}
Suppose the algorithm added the $(i+1)$th set $S$ during the $j$th pass. Consider the set of covered elements $C$  just before the algorithm added the set $S$.

We first consider the case where $j = 1$. Then, the algorithm only adds $S$ if 
\[|S \setminus C| \geq \frac{z}{k} \geq \frac{\opt}{ k} \geq \frac{c_{i}}{k} \geq \frac{c_i}{\alpha k} \ .\]

Now, we consider the case where $j > 1$. Note that just before the algorithm added $S$, there must exist a set $S'$ (which could be $S$) that had not  been already added where  $|S' \setminus C|\geq c_i/ k$. This follows because the optimal collection of $k$ sets covers at least $c_i$ elements that are currently uncovered and hence one of these sets must cover at least $c_i/k$ new elements. 
But since $S'$ had not already been added, we know that $S'$ was not added during the first $j-1$ passes and thus, $|S'\setminus C|< z/ (k \alpha^{j-2})$.
Therefore, 
\[
\frac{z}{k \alpha^{j-2}} > |S'\setminus C|\geq \frac{c_i}{k}
\] 
and in particular, $z/(k \alpha^{j-1}) > c_i/(k \alpha)$. Since the algorithm picked $S$, we have 
\[
a_{i+1}= |S\setminus C|\geq \frac{z} { k\alpha^{j-1} } \geq \frac{c_i}{k \alpha}
\] 
as required.\end{proof}

\begin{proof}[Proof of Lemma \ref{lem:core-1}]
It is immediate that the number of passes is  $O(\epsilon^{-1} )$.  The algorithm needs to store the sets $I$ and $C$. 
Since $|C| \leq z$, the total space is $O(k \log m + z \log n)$.

To argue about the approximation factor, we first prove by induction that we always have $c_i \leq \left( 1-\frac{1}{\alpha k} \right)^i \opt$ for $i \leq k'$. Trivially,  $c_0 \leq \paren{1-\frac{1}{\alpha k}}^0 \opt$. Suppose  $c_i \leq (1-\frac{1}{\alpha k})^i \opt$. Then, according to Lemma \ref{lem:3}, $a_{i+1} \geq {c_i}/({\alpha k})$. Thus,
\[
c_{i+1} = c_i - a_{i+1} \leq c_i - \frac{c_i}{\alpha k} = c_i  \left( 1-\frac{1}{ \alpha k} \right) \leq \opt \left( 1-\frac{1}{\alpha k} \right)^{i+1} .
\]

Suppose the final solution contains $k$ sets. Then 
\[
c_k \leq \left (1-\frac{1}{\alpha k} \right )^k \opt \leq e^{-1/\alpha} \opt \leq \paren{\frac{1}{e} + \epsilon} \opt.
\]
As a result, the final solution covers $b_k = \opt - c_k \geq \paren{1-1/e -\epsilon} \opt$ elements.

Suppose the collection of sets $\S$ chosen by the algorithm contains fewer than $k$ sets. We define $\tilde{S} := S \setminus \cover(\S)$ to be the set of elements in $S$ that are not covered by the final solution. For each set $S$ in the optimal solution $\O$, if $S$ is unpicked, then $|\tilde{S}| \leq z/(4 e k)$. Therefore,
\begin{align*}
\opt  = \left| \bigcup_{S \in \O } \left(S \cap \cover(\S) \right) ~ \right| + \left| \bigcup_{S\in \O \setminus \S} \tilde{S} ~ \right|   \leq \left| \cover( \S) \right| + \sum_{S\in \O \setminus \S} \left| \tilde{S}   \right|  & \leq \left| \cover( \S) \right| + \frac{z}{4e}  \\
& \leq \left| \cover( \S) \right| + \frac{\opt}{e} 
\ .
\end{align*}
Hence, $\left| \cover(\S) \right| \geq  (1-1/e)\opt$.
\end{proof}

Following the approach outlined in Section \ref{sec:magic}  we may assume $z=O(\epsilon^{-2} k \log m)$ and that $\opt \leq z\leq 4\opt.$

\begin{theorem}\label{thm:key-result-1}
There exists an $O(\epsilon^{-1})$-pass, $\tO(\epsilon^{-2} k)$ space algorithm that finds a $1-1/e-\epsilon$ approximation of  \MkSC with high probability.
\end{theorem}

\subsection{Removing Assumptions via Guessing, Sampling, and Sketching}
\label{sec:magic}
In this section, we address the fact that in the previous two sections we assumed a priori knowledge of a constant approximation of the maximum number of elements that could be covered and that this optimum was of size $O(\epsilon^{-2} k \log m)$.

Addressing both issues are interrelated and are based on a subsampling approach. The basic idea is to run the above algorithms on a new instance formed by removing  occurrences of certain elements in $[n]$ from all the input sets. The goal is to reduce the maximum coverage to $\min(n,O(\epsilon^{-2} k \log m))$ while ensuring that a good approximation in the subsampled instance corresponds to a good approximation in the original instance. In the rest of this section we will assume that $k=o(\epsilon^{2} n /\log m)$ since otherwise this bound is trivial.

In this section, we will need to use the  following Chernoff bound for limited independent random variables. 
\begin{theorem}[Schmidt  et al.~\cite{SchmidtSS95}]\label{thm:chernoff-hashing}
Let $X_1,\ldots, X_n$ be binary random variables. Let $X =\sum_{i=1}^n X_i$ and $\mu = \expec{X}$. Suppose $\mu \leq n/2$. If $X_i$ are $\lceil \gamma \mu \rceil$-wise independent, then
\[\prob{\left| X-\mu \right| \geq \gamma \mu} \leq \expo{-\lfloor \min(\gamma, \gamma^{2}) \cdot \mu/3 \rfloor}\ .\]
\end{theorem}

\paragraph{Subsampling.} Assume we know a value $v$ that satisfies $\opt/2 \leq v\leq \opt$. 
Let $c$ be some sufficiently large constant and set $\lambda=c \epsilon^{-2}  k \log m$.  Let $h:[n]\rightarrow \{0,1\}$ be drawn from a family of  $2 \lambda$-wise independent hash functions  where 
\[p := \prob{h(e) = 1} = {\lambda}/{v}. \]
The space to store $h$ is $\tO(\epsilon^{-2} k )$. For any set $S$ that is a subset of $[n]$, we  define 
\[S' :=  \{ e \in S: h(e) = 1 \}.\]
The next lemma and its corollary will allow us to argue that approximating the maximum coverage among the elements  $\{ e \in [n]: h(e) = 1 \}$ gives only a slightly weaker approximation of the maximum coverage among the original set of elements.
\begin{lemma} \label{lem:approximation-hashing}
With high probability\footnote{We consider $1-1/\poly(m)$ or $1-1/\poly(n)$ as high probability.}, for all collections of $k$ sets $S_1,\ldots,S_k$ in the stream,
$|S_1' \cup \ldots \cup S_k'|=|S_1 \cup \ldots \cup S_k|p \pm \epsilon vp \ .$
\end{lemma}
\begin{proof}
Fix any collection of $k$ sets $S_1,\ldots,S_k$. Let $D=|S_1 \cup \ldots \cup S_k| $ and $D'=|S_1' \cup \ldots \cup S_k'|$. We first observe that since $k=o(\epsilon^{2} n /\log m)$, we may assume that $\lambda =o( n)$.
\[\mu:=\expec{D'} = p D \leq p  \opt < 2 p v = 2\lambda \leq n/2.  \]

Appealing to the Chernoff bound with limited independence (Theorem \ref{thm:chernoff-hashing}) with the binary variables $X_i = 1$ if and only if $i \in S_1 \cup \ldots \cup S_k$ and $h(i)=1$, i.e, $D' = \sum_{i=1}^n X_i$, we have
\[
\prob{|D' -\mu| \geq \epsilon vp}
=
\prob{|D' -\mu| \geq \gamma Dp}\leq \exp \left (-\lfloor \min ( \gamma,\gamma^2) \cdot \mu/3\right \rfloor )
\]
where $\gamma=\epsilon v/D$ since the hash function is $\lceil \gamma \mu \rceil=\lceil \epsilon vp\rceil$-wise independent. But note that 
\begin{align*}
\exp \left(-\lfloor \min ( \gamma,\gamma^2) \cdot \frac{\mu}{3} \rfloor \right)
& =\exp \left( -\lfloor \min ( 1,\gamma) \cdot  \frac{\epsilon v p}{ 3} \rfloor \right) \\
&  \leq \exp \left( -\left \lfloor \frac{1}{2} \cdot \frac{c k \log m}{3} \right \rfloor \right) \leq \frac{1}{m^{10k}}
\end{align*}
where we use the fact that $\gamma=\epsilon v/D\geq \epsilon/2$ because $D\leq \opt \leq 2v$.
The lemma follows by taking the union bound over all ${m \choose k}$ collections of $k$ sets.
\end{proof}

In particular, the following corollary establishes that a $1/t$ approximation when restricted to elements in $\{ e \in [n]: h(e) = 1 \}$ yields a $(1/t-2\epsilon)$ approximation and at most $p\opt(1+\epsilon)=O(\epsilon^{-2} k \log m)$ of these elements can be covered by $k$ sets.

\begin{corollary}\label{cor:subsample-opt}
Let $\opt'$ be optimal number of elements that can be covered from $\{ e \in [n]: h(e) = 1 \}$. Then,
\[
p\opt(1+\epsilon) \geq
\opt'
\geq p\opt(1-\epsilon)
\]
Furthermore if $U_1,\ldots,U_k$ satisfies $|U_1' \cup \ldots \cup U_k'|\geq p\opt (1-\epsilon)/t$ for $t\geq 1$ then 
\[|U_1 \cup \ldots \cup U_k|\geq \opt \paren{\frac{1}{t}-2\epsilon} \ . \]
\end{corollary}
\begin{proof}
The fact that $\opt'
\geq p\opt(1-\epsilon)$ follows by applying Lemma \ref{lem:approximation-hashing} to the optimal solution. According to Lemma \ref{lem:approximation-hashing}, for all collections of $k$ sets $U_1,\ldots, U_k$, we have \[|U_1' \cup \ldots \cup U_k'|=|U_1 \cup \ldots \cup U_k|p \pm \epsilon vp\leq p\opt (1+\epsilon)\] which implies the first inequality. 

Now, suppose  $|U_1' \cup \ldots \cup U_k'|\geq p\opt (1-\epsilon)/t$. Since $|U_1' \cup \ldots \cup U_k'| - \epsilon vp  \leq |U_1 \cup \ldots \cup U_k|p  $, we deduce that 
$ |U_1 \cup \ldots \cup U_k|\geq \opt (1-\epsilon)/t-\epsilon v\geq \opt(1/t-2\epsilon)$.
\end{proof}

Hence, since  we know $v$ such that $\opt/2 \leq v\leq \opt$, then we know that 
\begin{align}  \label{eq:1}
(1-\epsilon) \lambda \leq \opt' \leq 2 (1+ \epsilon) \lambda 
\end{align}
with high probability according to  Corollary \ref{cor:subsample-opt}. Then, by setting $z = 2 (1+ \epsilon) \lambda$, we ensure that $ \opt' \leq z \leq 4 \opt'$.

\paragraph{Guessing $v$ and $F_0$ Sketching.} We still need to address how to compute $v$ such that $\opt/2 \leq v \leq \opt$. The natural approach is to make $\lceil \log_2 n \rceil$ guesses for $v$ corresponding to $1, 2, 4, 8 \ldots$ since one of these will be correct.\footnote{The number of guesses can be reduced to $\lceil \log_2 k \rceil$ if the size of the largest set is known since this gives a $k$ approximation of $\opt$. The size of the large set can be computed in one additional pass if necessary.} We then perform multiple parallel instantiations of the algorithm corresponding to each guess. This increases the space by a factor of $O(\log n)$.

But how do we determine which instantiation corresponds to the correct guess? The most expedient way to deal with this question is to sidestep the issue as follows. Instantiations corresponding to guesses that are too small may find it is possible to cover $\omega(\epsilon^{-2} k \log m)$ elements so we will terminate any instantiation as soon as it covers more than $O(\epsilon^{-2} k \log m)$ elements. Note that by Corollary \ref{cor:subsample-opt} and Equation \ref{eq:1}, we will not terminate the instantiation corresponding to the correct guess.

Among the instantiations that are not terminated we simply return the best solution. To find the best solution we want to estimate $|\cup_{i\in I} S_i|$, i.e., the coverage of the corresponding sets \emph{before} the subsampling. To compute this estimate in small space we can use the \emph{$F_0$-sketching} technique. For the purposes of our application, we can summarize the required result as follows:

\begin{theorem}[Cormode et al.~\cite{CormodeDIM03}]\label{thm:F0-approximation}
There exists an $\tO(\epsilon^{-2}\log \delta^{-1})$-space algorithm that, given a set $S\subseteq [n]$, can construct a data structure $\mathcal{M}(S)$, called an \emph{$F_0$ sketch} of $S$, that has the property that the number of distinct elements in a collection of sets $S_1, S_2, \ldots, S_r$ can be approximated up to a $1 + \epsilon$ factor with  probability at least $1-\delta$ given the collection of $F_0$ sketches $\mathcal{M}(S_1), \mathcal{M}(S_2), \ldots, \mathcal{M}(S_r)$.
\end{theorem}

For the algorithms in the previous section, we can maintain a sketch $\mathcal{M}(C)$ of the set of covered elements in $\tO(\epsilon^{-2}\log \delta^{-1})$ space and from this can estimate the desired coverage. We set $\delta \leftarrow \Theta(1/n \cdot \log n)$ so that coverages of all non-terminated instances are estimated up to a factor $(1+\epsilon)$ with high probability.

\subsection{Other Algorithmic Results}

In this final subsection, we briefly review some other algorithmic results for \MkSC, either with different trade-offs or for a ``budgeted'' version of the problem.

\subsubsection{$(1-\epsilon)$ approximation in one pass and $\tO(\epsilon^{-2}mk)$ space} In the previous subsection, we gave a single-pass $1-1/e-\epsilon$ approximation using $\tO(\epsilon^{-2} m)$ space. Here we observe that if we are permitted  $\tO(\epsilon^{-2} mk)$ space and unlimited post-processing time then a $1-\epsilon$ approximation can be achieved directly from the $F_0$ sketches. 

In particular, in one pass we construct the $F_0$ sketches of all $m$ sets, $\mathcal{M}({S}_1),\ldots,\mathcal{M} ({S}_m)$ where the failure probability of the sketches is set to  $\delta = {1}/({n m^k})$. Thus, at the end of the stream, one can $1+\epsilon$ approximate the coverage $|S_{i_1}\cup \ldots \cup  S_{i_k}|$ for each collection of $k$ sets $S_{i_1}, \ldots, S_{i_k}$ with probability at least $1- 1/({n m^k})$. Since there are at most ${m \choose k} \leq m^k$ collections of $k$ sets, appealing to the union bound, we guarantee that the coverages of all of the collections of $k$ sets are preserved up to a $1+\epsilon$ factor with probability at least $1-1/n$. The space to store the sketches is $\tO(\epsilon^{-2}mk)$. 

\begin{theorem}\label{thm:key-result-5}
There exists a single-pass, $\tO(\epsilon^{-2}mk)$-space algorithm that finds a $1-\epsilon$ approximation of \MkSC with high probability .
\end{theorem}

In comparison to the algorithm in Theorem \ref{thm:key-result-4}, the algorithm above is non-adaptive. It also uses less space in the case where $k$ is much smaller than  $\epsilon^{-1}$. 

\subsubsection{$(1/2-\epsilon)$ \label{sec:1/2-approx} approximation in one pass and $\tO(\epsilon^{-3}k)$ space}
We next observe that it is possible to achieve a $1/2-\epsilon$ approximation using a single pass and $\tO(\epsilon^{-3} k)$ space. 
Consider the following simple single-pass algorithm that uses an estimate $z$ of $\opt$ such that  $\opt \leq z \leq (1+\epsilon) \opt$. As with previous algorithms, the basic algorithm in this section also maintains $I\subseteq [m]$, $C\subseteq [n]$ where $I$ corresponds to the ID's of the (at most $k$) sets in the current solution and $C$ is the the union of the corresponding sets. The algorithm proceeds as follows:

\begin{enumerate}
\item Initialize $C=\emptyset$ and $I=\emptyset$.
\item For each set $S$ in the stream: 
\begin{enumerate}
\item If $|S\setminus C| \geq  {z/(2k)}$ and $|I|<k$ then $I\leftarrow I\cup \{ID(S)\}$ and $C\leftarrow C\cup S$.
\end{enumerate}
\end{enumerate}

The described algorithm is a $1/2-\epsilon$ approximation. To see this, if the solution consists of $k$ sets, then the final solution obviously covers at least $z/2\geq \opt/2$ elements. Now we consider the case in which the collection of sets $\S$ chosen by the algorithm contains fewer than $k$ sets. We define $\tilde{S} := S \setminus \cover(\S)$ to be the set of elements in $S$ that are not covered by the final solution.  For each set $S$ in the optimal solution $\O$,   if $S$ is unpicked, then $|\tilde{S}| \leq z/(2k)$. Therefore,
\begin{align*}
\opt  =\left| \bigcup_{S \in \O } \left(S \cap \cover(\S) \right) ~ \right| + \left| \bigcup_{S\in \O \setminus \S} \tilde{S} ~ \right|     \leq \left| \cover( \S) \right| + \sum_{S\in \O \setminus \S} \left| \tilde{S}   \right| &  \leq \left| \cover( \S) \right| + \frac{z}{2}  \\
 &  \leq \left| \cover( \S) \right| + \frac{\opt (1+\epsilon)}{2} 
\end{align*}
and thus $\left| \cover(\S) \right| \geq  \frac{1-\epsilon}{2}\opt$.

We note that the above algorithm uses $O(k \log m + z \log n)$ space but we can use an argument similar to that used in Section \ref{sec:magic} to reduce this to $\tO(\epsilon^{-3} k)$. The only difference is since we need $z$ such that $\opt' \leq z\leq (1+\epsilon)\opt'$ we will guess $v$ in powers of $1+\epsilon/4$ and set $\lambda=16c\epsilon^{-2} k \log m$.  Then Eq. \ref{eq:1} becomes $(1-\epsilon/4)\lambda \leq \opt' \leq (1+\epsilon/4)^2 \lambda$ and hence $z=(1+\epsilon/4)^2 \lambda$ is a sufficiently good estimate.

\begin{theorem} \label{thm:key-result-3}
There exists a single-pass,  $\tO(\epsilon^{-3} k)$ space algorithm that finds a  $1/2-\epsilon$ approximation of \MkSC with high probability.
\end{theorem}

\subsubsection{Group Cardinality Constraint} \label{sec:group-cardinality}

In this version, we consider a version of the problem where each set belongs to a group amongst $\ell$ groups $G_1,G_2,\ldots,G_
\ell$ and we are allowed to pick at most $k_1$ sets from group $G_1$, $k_2$ sets from group $G_2$, and so on. This is also known as the partition matroid constraint. A $1-1/e$ approximation is possible for the offline version of this problem via linear programming \cite{AgeevS04,Srinivasan01}. Furthermore, Chekuri and Kumar showed that the offline greedy algorithm is guaranteed to return a $1/2$ approximation \cite{ChekuriK04}. 

\paragraph{Single-pass algorithm.} We first observe that by simply applying the previous $1/2-\epsilon$ approximation algorithm for each group and returning the best solution, we obtain a $(1/2-\epsilon)/\ell$ approximation. The main idea for the improved algorithm is to set a threshold for when to add a set that depends on the group to which this set belongs.

We now present an algorithm that returns a $1/(\ell+1) - \epsilon$ approximation which improves upon \cite{ChakrabartiK14,ChekuriGQ15} for the case $\ell = 2$. The basic algorithm maintains the sets $I_i$ for each $i = 1,2,\ldots,\ell$ where $I_i$ corresponds to the IDs of the sets from group $G_i$ that are in the current solution. Similar to previous algorithms, $C$ is used to keep track of the current coverage. Finally, the algorithm also uses an estimate $z$ of $\opt$ such that  $\opt \leq z \leq (1+\epsilon) \opt$. The detailed algorithm proceeds as follows:
 
\begin{enumerate}
\item Initialize $C=\emptyset$ and $I_i =\emptyset$ for each $i = 1,\ldots,\ell$.
\item For each set $S \in G_i$ in the stream: if

\[
|S\setminus C| \geq  {\frac{z}{(\ell+1)k_i}} \text{ and } |I_i|<k_i,
\] 
then $I_i \leftarrow I_i\cup \{ID(S)\}  ~~\mbox{ and }~~ C\leftarrow C\cup S \ . $
\end{enumerate}

If there exists a group $G_i$ in which $k_i$ sets are selected,  then it is clear that the solution covers  at least $z/(\ell+1)$ elements.
On the other hand, suppose that for all groups $G_i$, fewer than $k_i$ sets are selected. As before, $\S$ and $\cover(\S)$ are the collection of sets in the solution and their union respectively. Again, we define $\tilde{S} := S \setminus \cover(\S)$. Furthermore, let $\O_i$ denote the sets in $G_i$ that are also in the optimal solution, i.e., $\O_i = \O \cap G_i$. We have 
\begin{align*}
\opt  =\left| \bigcup_{S \in \O } \left(S \cap \cover(\S) \right) ~ \right| + \left| \bigcup_{S\in \O \setminus \S} \tilde{S} ~ \right|  &  \leq \size{\cover(\S)} + \sum_{i=1}^\ell \sum_{S \in \O_i \setminus \S} \left| ~ \tilde{S} ~ \right| \\ 
 & \leq \size{\cover(\S)} + \sum_{i=1}^\ell \sum_{S \in \O_i \setminus \S} \frac{z}{(\ell+1)k_i}  \leq  \size{\cover(\S)} + \frac{\ell \cdot z}{\ell+1} ~.
\end{align*}
Therefore,
\[
\size{\cover(\S)} \geq  \opt - \frac{\ell \cdot z}{\ell+1} \geq  \paren{ \frac{1}{\ell+1}-\epsilon } \opt ~.
\]
Let $k = k_1 + k_2 + \ldots + k_\ell$.  The above algorithm uses $O(k \log m + z \log n)$ space but we can use an argument similar to that used in Sections \ref{sec:magic} and \ref{sec:1/2-approx} to reduce this to $\tO(\epsilon^{-3} k)$. We summarize the result as the following theorem.

\begin{theorem}\label{thm:sub-result-1}
There exists a single-pass, $\tO(\epsilon^{-3} k)$ space algorithm that finds a $1/(\ell+1) - \epsilon$ approximation of group cardinality constraint \MkSC with high probability.
\end{theorem}

We note that for single pass, the algorithm of \cite{ChakrabartiK14,ChekuriGQ15} combining with our framework in Sections \ref{sec:magic} yields a $1/4-\epsilon$ approximation for matroid constraints. The algorithm above gives a better approximation for partition matroid constraint when the number of groups $\ell = 2$.

\begin{theorem}
There exists a single-pass, $\tO(\epsilon^{-2} k^2)$ space algorithm that finds a $1/4-\epsilon$ approximation of group cardinality constraint \MkSC with high probability.
\end{theorem}

\paragraph{Multiple-pass algorithm.}  Next, we demonstrate a $1/2-\epsilon$ approximation that uses $O(\epsilon^{-1} \log (k/\epsilon))$ passes. The idea is similar to the algorithm in Section \ref{sec:thesholding-greedy} where we pick a set if its contribution is above a threshold. We decrease the threshold by a factor $(1+\epsilon)$ after each pass. The main difference is to not pick a set if that violates the group constraint. Here, we assume that $\opt \leq z \leq 4\opt$. The detailed algorithm proceeds as follows:

\begin{enumerate}
\item Initialize $C \leftarrow \emptyset$ and $I_i =\emptyset$ for each $i = 1,\ldots,\ell$.
\item For $j=1$ to $\roundup{\log_{\alpha}(10 \cdot k/\epsilon)}$ where $\alpha = 1+\epsilon$
    \begin{enumerate}
        \item Make a pass over the stream. For each set $S \in G_i$:
        \begin{enumerate}
        		\item If $\size{S \setminus C} \geq z/ \alpha^j$ and $\size{I_i} < k_i$, then $I_i \leftarrow I_i\cup \{ID(S)\} ~~\mbox{ and }~~ C\leftarrow C\cup S \ .$
         \end{enumerate}
    \end{enumerate}
\end{enumerate}

Recall that $k = k_1+\ldots+k_\ell$. Suppose the algorithm picks $k$ sets $S_1,S_2,\ldots,S_k$ in that order. If the algorithm picks fewer than $k$ sets, at the end, we could simply add dummy empty sets to the solution; thus, we can assume that the algorithm picks exactly $k_i$ sets from each group $G_i$. Consider an optimal solution $\O = \{ O_1,\ldots,O_k \}$ and a bijection $\pi : [k] \rightarrow [k]$ that satisfies the following: 
\begin{itemize}
\item If $\pi(i) = j$, then $S_i$ and $O_j$ belong to the same group. 
\item If $S_i$ is $O_j$, then $\pi(i) = j$.
\end{itemize}
When $S_i \in G_t$ was picked in the $j$th iteration for $j > 1$, by the second property of $\pi$, we deduce that $O_{\pi(i)}$ had not been picked in the $(j-1)$th iteration.  Furthermore, the first property of $\pi$ ensures that since we picked $S_i$ in the $j$th iteration, we know that $O_{\pi(i)}$ was available to pick in the $(j-1)$th iteration; however, its contribution was smaller than $z/(k \alpha^{j-1})$. Let $\hat{S}_i := S_i \setminus (S_1 \cup \ldots S_{i-1})$ and $\tilde{O}_i := O_i \setminus \cover(\S)$.

Therefore,
\begin{align*}
\size{\tilde{O}_{\pi(i)}} < \frac{z}{\alpha^{j-1}} = \alpha \cdot \frac{z}{\alpha^j} \leq \alpha \size{\hat{S}_i} ~.
\end{align*}

In the case $j=1$, obviously, $\size{\hat{S_i}} \geq z/\alpha \geq \size{\tilde{O}_{\pi(i)}} \cdot 1/\alpha$. 

Finally, in the case that $S_i$ is a dummy set that was added at the end, then $O_{\pi(i)}$ was not picked during the last pass (even though it was available to pick). Hence $\size{\tilde{O}_{\pi(i)} } \leq \epsilon z/(10k)$. Suppose the algorithm picked $y$ sets $S_1,\ldots,S_y$ that are not dummy sets. We have 
\begin{align*}
\size{\cover(\O)} - \size{\cover(\S)} & \leq \sum_{i=1}^y \size{O_{\pi(i)} \setminus \cover(\S)} + \sum_{i=y+1}^k \size{O_{\pi(i)} \setminus \cover(\S)} \\
& \leq  \sum_{i=1}^y \alpha \size{\hat{S}_i}   + \epsilon \size{\cover(\O)} \\
(2+\epsilon) \size{\cover(\S)} & \geq (1-\epsilon) \size{\cover(\O)} ~.
\end{align*}

Therefore, $\size{\cover(\S)} \geq (1-\epsilon)\opt/(2+\epsilon)$. Repeating the subsampling argument in Section  \ref{sec:magic}, we have the following.

\begin{theorem}\label{thm:sub-result-2}
There exists a $\tO(\epsilon^{-2} k)$ space algorithm that finds a $1/2 -\epsilon$ approximation of group cardinality constraint \MkSC in $O(\epsilon^{-1} \log (k/\epsilon))$ passes with high probability .
\end{theorem}
%
%
%

\subsubsection{Budgeted Maximum Coverage}

In this variation, each set $S$ has a cost $w_S \in [0,L]$. The problem asks to find the collection of sets whose total cost is at most $L$ that covers the most number of distinct elements. For $I\subseteq [n]$, we use $w(I)$ to denote $\sum_{i\in I}w_{S_i}$.

We present the algorithm assuming knowledge of an estimate $z$ such that $\opt \leq z \leq (1+\epsilon)\opt$; this assumption can be removed by running the algorithm for guesses $1, (1+\epsilon), (1+\epsilon)^2, \ldots$ for $z$ and returning the best solution found. The basic algorithm maintains $I\subseteq [m]$, $C\subseteq [n]$ where $I$ corresponds to the ID's of the (at most $k$) sets in the current solution and $C$ is the the union of the corresponding sets. The algorithm proceeds as follows:

\begin{enumerate}
\item Initialize $C=\emptyset$ and $I=\emptyset$
\item For each set $S$ in the stream:  
\begin{enumerate}
\item If
\[
|S\setminus C| \geq \frac{2z}{3} \cdot \frac{w_S}{L}~,
\] 
then:
\begin{enumerate}
\item If $w(I)+w_S> L$: Terminate and return:
\[
I\leftarrow \begin{cases}
I & \mbox{ if } |C|\geq |S| \\
\{ID(S)\} & \mbox{ if } |C|< |S| \\
\end{cases}
\]
\item $I\leftarrow I\cup \{ID(S)\}$ and $C\leftarrow C\cup S$.
\end{enumerate}
\end{enumerate}
\end{enumerate}

\begin{lemma}
If the clause in line $2ai$ is never satisfied, then the algorithm returns a $1/3-\epsilon$ approximation. 
\end{lemma}
\begin{proof}
Suppose the collection of sets chosen by the algorithm is $\S$. 
As before, we define $\tilde{S} := S \setminus \cover(\S)$ to be the set of elements in $S$ that are not covered by the final solution. 
For each set $S$ in the optimal solution $\O$, if $S$ is unpicked, then $|\tilde{S}| \leq 2z/3 \cdot w_S/L$. Therefore,
\begin{align*}
\opt  = \left| \bigcup_{S \in \O } \left(S \cap \cover(\S) \right) ~ \right| + \left| \bigcup_{S\in \O \setminus \S} \tilde{S} ~ \right|    \leq \left| \cover( \S) \right| + \sum_{S\in \O \setminus \S} \left| \tilde{S}   \right|  & \leq \left| \cover( \S) \right| + \frac{2z}{3} \\
& \leq \left| \cover( \S) \right| + \frac{2 \opt (1+\epsilon)}{3} 
\ ,
\end{align*}
and thus $\left| \cover(\S) \right| \geq  \frac{1-2\epsilon}{3}\opt$.
\end{proof}

\begin{lemma}
If the clause in line $2ai$ is satisfied at some point, then the algorithm returns a $1/3$ approximation. 
\end{lemma}
\begin{proof}
Suppose the clause is satisfied when the set $S$ is being considered. Then
\[ 
\left| S \setminus C\right| + \left|C \right| \geq \frac{2z}{3}  \cdot \frac{w_S + w(I) }{L} \geq \frac{2z}{3}
\]
where we used the fact that $w_S + w(I) > L$. The claim then follows immediately. 
\end{proof}

The algorithm needs to store the IDs of the sets in the solution as well as the current coverage $C$. Therefore, it uses $\tO(m+n)$ space.

\begin{theorem}\label{theorem:budget}
There exists a single-pass, $\tO(\epsilon^{-1} (m+n))$-space algorithm that finds a $1/3-\epsilon$ approximation of budgeted \MkSC.
\end{theorem}


\section{Algorithms for Maximum $k$-Vertex Coverage}

In this section, we present algorithms for the maximum $k$-vertex coverage problem. We present our results in terms of hypergraphs for full generality. The generalization to hypergraphs can also be thought of as a natural ``hitting set" variant of maximum coverage, i.e., the stream consists of a sequence of sets and we want to pick $k$ elements in such a way to maximize the number of sets that include a picked element.

\paragraph{Notation.} Given a hypergraph $G$ and a subset of nodes $S$, we define $\coverage_G(S)$ to be the number of edges that contain at least one node in $S$. Recall that the maximum $k$-vertex coverage problem is to approximate the maximum value of $\coverage_G(S)$ over all sets $S$ containing $k$ nodes. We use $E_G$ and $V_G$ to denote the set of edges and nodes of the hypergraph $G$ respectively. 

The size of a cut $(S,V \setminus S)$ in a hypergraph $G$, denoted as $\delta_G(S)$, is defined as the number of hyperedges that contain at least one node in both $S$ and $V \setminus S$. In the case that $G$ is weighted, $\delta_G(S)$ denotes the total weight of the cut. 
A core idea to our approach is to use \emph{hypergraph sparsification}:

\begin{definition}[$\epsilon$-sparsifier]
Given a hypergraph $G=(V,E)$, we say that a weighted subgraph $H =(V,E')$ is an $\epsilon$-sparsifier for $G$ if for all $S\subseteq V$, $\delta_G(S)\approx_\epsilon \delta_{H}(S)$. 
\end{definition}

Any graph on $N$ nodes has an $\epsilon$-sparsifier with only $\tilde{O}(\epsilon^{-2} N)$ edges \cite{SpielmanT11}. Similarly, any hypergraph  in which the maximum size of the hyperedges is bounded by $d$ (rank $d$ hypergraphs) has an $\epsilon$-sparsifier with only $\tilde{O}(\epsilon^{-2} d N)$ edges.
Furthermore, an $\epsilon$-sparsifier can be constructed in the dynamic graph stream model using one pass and $\tO(\epsilon^{-2} d N)$ space \cite{GuhaMT15,KapralovLMMS14}. 

First, we show that it is possible to approximate all the coverages by constructing a sparsifier of a slightly modified graph. In particular, we construct the sparsifier $H$ of the graph $G'$ with an extra node $v$, i.e., $V_{G'} = V_G \cup \{  v \}$, and for every hyperedge $e \in E_G$, we put the hyperedge $e \cup \{ v \}$ in $E_{G'}$. It is easy to see that for all $S$ that is a subset of $V_G$, we have $\cover_G(S) = \delta_{G'}(S).$ Therefore, it is immediate that we could $1+\epsilon$ approximate all the coverages in $G$ by constructing the sparsifer of $G'$. 
\begin{theorem}\label{thm:vertex-cover-general-2}
There exists a single-pass, $\tO(\epsilon^{-2} d N)$-space algorithm that finds a $1 -\epsilon$ approximation of \MkVC of rank $d$ hypergraphs with high probability.
\end{theorem} 

The above theorem assumes unbounded post-processing time. If $k$ is constant, the post-processing will be polynomial. For larger $k$, if we still require polynomial running time then, after constructing the $\epsilon$-sparsifier $H$, we could either use the $(1-(1-1/d)^d)$ approximation algorithm via linear programming \cite{AgeevS99} or the  folklore $(1-1/e)$ approximation greedy algorithm.
\subsection{Algorithm for Near-Regular Hypergraphs}
In this subsection, we show that it is possible to reduce the space used to $\tilde{O}(\epsilon^{-3} dk)$ in the case of hypergraphs that are regular or nearly regular. Define $\kappa \leq 1$ to be the ratio between the smallest degree and the largest degree; for a regular hypergraph $\kappa=1$. 
We show that a $(\kappa-\epsilon)$ approximation is possible using $\tO(\epsilon^{-3}d k)$ space for rank $d$ hypergraphs. This also implies a $(1-\epsilon)$ approximation for regular hypergraphs. Let $t_1$ and $t_2$ be the minimum and maximum degree of a node in $G$.
\begin{theorem} \label{thm:vertex-cover-bounded-2}
There exists a single-pass, $\tO(\epsilon^{-3} d k)$-space algorithm that finds a $(\kappa-\epsilon)$ approximation of \MkVC of hypergraphs of rank $d$ with high probability .
\end{theorem}

\begin{proof}
Suppose we uniformly sample a set $S$ of $k$ nodes. Let $L_S(y)=\max(0,|y\cap S|-1)$.  Then the  coverage of $S$ satisfies
\begin{align*}
\cover_G(S)  
= \sum_{y\in E_G} I[S\cap y\neq \emptyset]
= \sum_{y\in E_G} \left ( |S\cap y| - L_S(y) \right )
\geq kt_1 -  \sum_{y \in E_G} L_S(y)
 \ . 
\end{align*}
where the last inequality follows since every node in $S$ covers at least $t_1$ hyperedges. 

Let $\xi_y(j)$ denote the event that $j$ nodes in the hyperedge $y$ are in $S$ and let $|y|$ denote the number of nodes in $y$.
We  have
\begin{align*}
\expec{L_S(y) } & = \sum_{j=1}^{|y|} (j-1) \prob{\xi_y(j)}  = \left( \sum_{j=0}^{|y|}  j \prob{\xi_y(j)} \right)  - 1 + \prob{\xi_y(0)} ~.
& \end{align*}

The sum $\sum_{j=0}^{|y|}  j \prob{\xi_y(j)}$ is the expected value of the hypergeometric distribution and therefore it evaluates to $|y|k/N$. Furthermore, 
\begin{align*}
\prob{\xi_y(0)} = \prod_{i=0}^{k-1} \left( 1-\frac{|y|}{N-i} \right)
 & \leq  \left( 1-\frac{|y|}{N} \right)^{k} \\
 & \leq \exp \left(-\frac{k |y|}{N} \right)  \leq 1-  \frac{k |y|}{N} + \frac{1}{2}\left( \frac{k |y|}{N} \right)^2.
\end{align*}

The last inequality follows from taking the first three terms of the Taylor's expansion. Hence,  
\begin{align*}
\expec{L_S(y) }  \leq \frac{k |y|}{N} -1 + 1-  \frac{k |y|}{N} + \frac{1}{2}\left( \frac{k |y|}{N} \right)^2 = \frac{1}{2}\left( \frac{k |y|}{N} \right)^2. 
\end{align*}

Hence, if $N \geq 4 kd/\epsilon$, then 
\begin{align*}
\sum_{y \in E_G}\expec{L_S(y) }  \leq \frac{1}{2} \sum_{y \in E_G} \left( \frac{k |y|}{N} \right)^2   \leq \frac{1}{2}  d \left( \frac{k}{N} \right)^2 \sum_{y \in E_G} |y| & \leq \frac{1}{2}  d \left( \frac{k}{N} \right)^2 N t_2  \\
& \leq \frac{1}{8}  \epsilon k t_2 ~. 
\end{align*}

By an application of Markov's inequality, \[\prob{{\sum_{y \in E_G}L_S(y) }  \geq \epsilon k t_2} \leq 1/8 \ .\] Thus, if we sample $O(\log N)$ sets of $k$ nodes in parallel, with high probability, there is a sample set $S$ of $k$ nodes satisfying ${\sum_{y \in E_G}L_S(y) }  \leq \epsilon k t_2$ which implies that $\cover_G(S) \geq kt_1 - \epsilon kt_2  \geq (\kappa-\epsilon)\opt$.  If $N \leq 4 kd /\epsilon$, we simply construct the sparsifier of $G'$ as described above to achieve a $1-\epsilon$ approximation.
\end{proof}


\section{Lower Bounds}
In this section, we prove space lower bounds for data stream algorithms that approximate  \MkSC or \MkVC. In particular, these imply that improving over an $(1-1/e)$ approximation of \MkSC with constant passes and constant $k$ requires $\Omega(m)$ space. Recall that, still assuming $k$ is constant, we designed a constant-pass algorithm that returned a $(1-1/e-\epsilon)$ approximation using $\tO(\epsilon^{-2} k)$ space. For constant $k$, we also show that improving over a $\kappa$ approximation (where $\kappa$ is the ratio between the lowest degree and the highest degree) for \MkVC requires $\Omega(N \kappa^3)$ space. Our algorithm returned a $\kappa-\epsilon$ approximation using $\tO(\epsilon^{-3} k)$ space.
\paragraph{Approach.} We prove both bounds by a reduction from $r$-player set-disjointness in communication complexity.   In this problem, there are $r$ players where the  $i$th player has a set $S_{i} \subseteq [u]$. It is promised that exactly one of the following two cases happens.
\begin{itemize}
\item Case 1 (NO instance): All the sets are pairwise disjoint.
\item Case 2 (YES instance): There is a unique element $e \in [u]$ such that $e \in S_i$ for all $i \in [r]$.
\end{itemize} 

The goal of the communication problem is the $r$th player answers whether the input is a YES instance or a NO instance correctly with probability at least 0.9. We shall denote this problem by $\Disj_r(u)$.

 The communication complexity of the above problem in $p$-round,  one-way model (where each round consists of player 1 sending a message to player 2, then player 2 sending a message to player 3 and so on) is $\Omega(u/r)$ \cite{ChakrabartiKS03} even if the players may use public randomness. This implies that in any randomized communication protocol, the maximum message sent by a player contains $\Omega(u/(pr^2))$ bits. Without loss of generality, we could assume that $|S_1 \cup S_2 \cup \ldots \cup S_r| \geq u/4$ via a padding argument.

\begin{theorem}\label{thm:lower-bound-1}
Assuming $n = \Omega(\epsilon^{-2} k \log m)$, any constant-pass algorithm that finds a $(1+\epsilon)(1-(1-1/k)^k) $ approximation of  \MkSC  with probability at least 0.99 requires $\Omega(m/k^2)$ space even when all the sets have the same size.
\end{theorem}
\begin{proof}
 Our proof is a reduction from $\Disj_k(m)$.
Consider a sufficiently large $n$ where  $k$ divides $n$. For each $i \in [m]$, let $\mathcal{P}_i$ be a random partition of $[n]$ into $k$ sets $V^i_1,\ldots,V^i_k$ of equal size. Each partition is chosen independently and the players agree on these partitions using public randomness before receiving the input. 

For each player $j$, if $i \in S_j$, then she puts $V^i_j$ in the stream. According to the aforementioned assumption, the stream consists of $\Theta(m)$ sets.

If the input is a NO instance, then for each $i \in [m]$, there is at most one set $V^i_j$ in the stream. Hence, the stream consists of independent random sets of size $n/k$. Therefore, for each $e \in [n]$ and any $k$ sets $V^{i_1}_{j_1},\ldots, V^{i_k}_{j_k}$ in the stream,
$
\prob{e \in V^{i_1}_{j_1} \cup \ldots \cup V^{i_k}_{j_k}} = 1 - (1-1/k)^k.
$
By an application of Chernoff bound for negatively correlated boolean random variables \cite{PanconesiS97}, 
\begin{align*}
 & \prob{\left| | V^{i_1}_{j_1} \cup \ldots \cup V^{i_k}_{j_k}| -  \left(1- \left(1-\frac{1}{k} \right)^k \right) n \right| > \epsilon \left(1- \left(1-\frac{1}{k}\right)^k \right) n  }  \\
 &\leq 3 \expo{-\epsilon^2 \left(1-\left(1-\frac{1}{k} \right)^k\right) \frac{n}{3}}  \\
 & \leq  3 \expo{-\epsilon^2 (1-{1}/{e}) n/3} \leq  \frac{1}{m^{10+k}} ~.
\end{align*}

The last inequality holds when $n$ is a sufficiently large multiple of $k \epsilon^{-2} \log m$. Therefore, the maximum coverage in this case is at most $(1+\epsilon)(1-(1-1/k)^k)n$ with probability at least $1-1/m^{10}$ by taking the union bound over all ${m \choose k} \leq m^k$ possible $k$ sets.

If the input is a YES instance, then clearly, the maximum coverage is $n$. This is because  there exists $i \in [m]$ such that $i \in S_1 \cap \ldots \cap S_k$ and  therefore $V^i_1,\ldots,V^i_k$ are in the stream.

Therefore, any constant pass and $O(s)$-space algorithm that finds a $(1+2\epsilon)(1-(1-1/k)^k)$ approximation of the maximum coverage with probability at least 0.99 implies a protocol to solve the $k$-party disjointness problem using $O(s)$ bits of communication. Thus, $s = \Omega(m/k^2)$ as required. 
\end{proof}

Consider the sets $S_1,\ldots,S_r \subseteq [u]$ that satisfy the unique intersection promise as in $\Disj_r(u)$. Let $X$ be the $r$ by $u$ matrix in which the row $X_{i}$ is the characteristic vector of $S_i$. Suppose there are $r' = \Omega(r^2)$ players. Chakrabarti \etal \cite{ChakrabartiCM11} showed that if each entry of $X$ is given to a unique player and the order in which the entries are given to the players is random, then the players need to use $\Omega(u/r)$ bits of communication to tell whether the sets is a YES instance or a NO instance  with probability at least 0.9. Thus, in any randomized protocol, the maximum message sent by a player contains $\Omega(u/r^3)$ bits. Hence, using the same reduction and assuming constant $k$, we show that the same lower bound holds even for random order stream.

\begin{theorem}\label{thm:lower-bound-3}
Assuming $n = \Omega(\epsilon^{-2} k \log m)$, any constant-pass algorithm that finds a $(1+\epsilon)(1-(1-1/k)^k) $ approximation of  \MkSC  with probability at least 0.99 requires $\Omega(m/k^3)$ space even when all the sets have the same size and arrive in random order.
\end{theorem}

Next, we prove a lower bound for the $k$-vertex coverage problem for graphs where the ratio between the minimum degree and the maximum degree is at least $\kappa$. We show that for constant $k$, beating $\kappa$ approximation for constant $\kappa$ requires  $\Omega(N)$ space.

Since $\kappa$ can be smaller than any constant, this also establishes that $\Omega(N)$ space is required for any constant approximation of \MkVC.

\begin{theorem} \label{thm:lower-bound-2}
For $\epsilon>0$, any constant-pass algorithm that finds a $(\kappa+\epsilon)$ approximation of  \MkVC with probability at least 0.99 requires $\Omega(N\kappa^3/k)$ space.
\end{theorem}
\begin{proof}
Initially, assume $k=1$. We consider the multi-party set disjointness problem $\Disj_t(N')$ where $t=1/\kappa$ and $N'=N/t$. Here, there are $t$ players and the input sets are subsets of $[N']$. We consider a bipartite graph where the set of possible nodes are $L\cup R$ where $L=\{u_i\}_{i\in [N']}$ and $R=\{v_{i,j}\}_{i\in [N'], j\in [t]}$. Note that this graph has $(t+1)N'=\Theta(N)$ nodes. However we only consider a node to exist if the stream contains an edge incident to that node. 

The $j$-th player defines a set of edges on this graph based on their set $S_j$ as follows. If $i\in S_j$, she puts the edge between $u_i$ and $v_{i,j}$. If $S_1, \ldots, S_t$ is a YES instance, then there must be a node $u_i$ that has degree $t$. If $A$ is a NO instance, then every node in the graph has degree at most 1. Hence the ratio of minimum degree to maximum degree is at least $1/t=\kappa$ as required.

Thus, for $k=1$, a $1/t$ approximation with probability at least 0.99 on a graph of $N$ nodes implies a protocol to solve $\Disj_t(N')$. Therefore, the algorithm requires $\Omega(N\kappa^3)$ space. For general $k$, we make $k$ copies of the above construction to deduce the lower bound $\Omega(N\kappa^3/k)$. 
\end{proof}

\paragraph*{Acknowledgements.} We thank Sagar Kale for discussions of related work.

\bibliographystyle{plain} 
\bibliography{coverage}

\end{document}